\documentclass[conference]{IEEEtran}
\usepackage{multirow}
\usepackage{float}
\usepackage{url}
\usepackage{amsmath}
\usepackage{array}
\usepackage{amssymb}
\usepackage{amsfonts}
\usepackage{threeparttable}
\newfloat{equationbottom}{b}{lop}
\newtheorem{proposition}{Proposition}

\newtheorem{definition}{Definition}
\newtheorem{theorem}{Theorem}
\newtheorem{lemma}{Lemma}
\newtheorem{example}{Example}

\makeatletter
\renewcommand*\env@matrix[1][c]{\hskip -\arraycolsep
  \let\@ifnextchar\new@ifnextchar
  \array{*\c@MaxMatrixCols #1}}
\makeatother
\title{On the Sphere Decoding Complexity of STBCs for Asymmetric MIMO Systems}
\author{
\authorblockN{Lakshmi Prasad Natarajan, K.~Pavan Srinath and B.~Sundar Rajan}
\authorblockA{Dept. of ECE, IISc, Bangalore 560012, India\\
Email: \{nlp,pavan,bsrajan\}@ece.iisc.ernet.in\\}
}
\date{\today}
\begin{document}
\maketitle
\thispagestyle{empty}

\begin{abstract}
In the landmark paper~\cite{HaH} by Hassibi and Hochwald, it is claimed without proof that the upper triangular matrix $\textbf{R}$ encountered during the sphere decoding of any linear dispersion code is full-ranked whenever the rate of the code is less than the minimum of the number of transmit and receive antennas. In this paper, we show that this claim is true only when the number of receive antennas is at least as much as the number of transmit antennas. We also show that all known families of high rate (rate greater than $1$ complex symbol per channel use) multigroup ML decodable codes have rank-deficient $\textbf{R}$ matrix even when the criterion on rate is satisfied, and that this rank-deficiency problem arises only in asymmetric MIMO with number of receive antennas less than the number of transmit antennas. Unlike the codes with full-rank $\textbf{R}$ matrix, the average sphere decoding complexity of the STBCs whose $\textbf{R}$ matrix is rank-deficient is polynomial in the constellation size. We derive the sphere decoding complexity of most of the known high rate multigroup ML decodable codes and show that for each code, the complexity is a decreasing function of the number of receive antennas.
\end{abstract}

\section{Introduction} \label{sec1}

\subsection{System Model and Definitions}

We consider Space-Time Block Codes (STBCs) for an $N$ transmit antenna, $M$ receive antenna, quasi-static MIMO channel ($N \times M$ MIMO system) with Rayleigh flat fading. The system can be modeled as 
\begin{equation}\label{sys_model}
 \textbf{Y = XH + W},
\end{equation}
where $\textbf{X}$ is the \mbox{$T \times N$} codeword matrix transmitted over $T$ channel uses, $\textbf{Y}$ is the \mbox{$T \times M$} received matrix, $\textbf{H}$ is the \mbox{$N \times M$} channel matrix and the \mbox{$T \times M$} matrix $\textbf{W}$ is the additive noise at the receiver. The entries of $\textbf{H}$ and $\textbf{W}$ are i.i.d. zero mean, unit variance, circularly symmetric complex Gaussian random variables. 

\begin{definition}
(STBC) An STBC $\mathcal{C}$ encoding $K$ real independent information symbols, denoted by $x_i$, $i=1,\cdots,K$, is a set of complex matrices given by 
\begin{equation}\label{stbc}
\mathcal{C}=\left\{ \sum_{i=1}^{K}{x_i\textbf{A}_i} \bigg|[x_1,\dots,x_K]^T \in \mathcal{A} \right\},
\end{equation}
where the $T \times N$ complex matrices \mbox{$\textbf{A}_1,\dots,\textbf{A}_K$}, which are called \emph{linear dispersion} or \emph{weight matrices}, are linearly independent over the real field $\mathbb{R}$~\cite{HaH},~\cite{RaR}, and the finite set \mbox{$\mathcal{A} \subset \mathbb{R}^K$} is called the \emph{signal set}.
\end{definition}

\begin{definition}
 (Code Rate) The rate of an STBC is the average number of information symbols transmitted in each channel use. For the STBC given by \eqref{stbc}, the code rate is \mbox{$R=\frac{K}{T}$} real symbols per channel use, or \mbox{$R=\frac{K}{2T}$} complex symbols per channel use (cspcu).
\end{definition}

The linear independence of the weight matrices in the definition of an STBC implies that $R \leq N$. {\it Throughout this paper, unless otherwise specified, the code rate is taken to be in terms of complex symbols per channel use.}

Generally, the signal set $\mathcal{A}$ is chosen in such a way that the STBC $\mathcal{C}$ has full-diversity and large coding gain. In most cases $\mathcal{A}$ is chosen to be a subset of ${\bf \Theta}\mathbb{Z}^K$, where ${\bf \Theta} \in \mathbb{R}^{K \times K}$ is a full-ranked matrix. One such instance is when the symbols are partitioned into multiple encoding groups, and each group of symbols is encoded independently of other groups using a lattice constellation, such as in Clifford Unitary Weight Designs~\cite{RaR} (in which case ${\bf \Theta}$ is an orthogonal matrix), Quasi-Orthogonal STBCs~\cite{WWX} and Coordinate Interleaved Orthogonal Designs~\cite{KhR}. There are also instances where the real symbols are encoded independently using regular PAM constellations of possibly different minimum distances \cite{NaR2}, in which case ${\bf \Theta}$ is a diagonal matrix with positive entries.

For a complex matrix $\textbf{A}$, let its real and imaginary components be denoted by $\textbf{A}_I$ and $\textbf{A}_Q$, respectively. Let $vec(\textbf{A})$ denote the complex vector obtained by stacking the columns of $\textbf{A}$ one below the other and
\begin{equation*}
\widetilde{vec}(\textbf{A}) \triangleq [vec(\textbf{A}_I)^T ~ vec(\textbf{A}_Q)^T]^T. 
\end{equation*}
Now, the system model given by \eqref{sys_model} can be expressed as 
\begin{equation*}
\textbf{y} = \widetilde{vec}(\textbf{Y}) = \textbf{Gx} + \textbf{w},
\end{equation*}
 where \mbox{$\textbf{x} = [x_1,\dots,x_K]^T$}, \mbox{$\textbf{w}=\widetilde{vec}(\textbf{W})$} and the \emph{equivalent channel matrix} $\textbf{G} \in \mathbb{R}^{2MT \times K}$ is given by
\begin{equation*}
	\textbf{G} = [\widetilde{vec}(\textbf{A}_1\textbf{H})~\widetilde{vec}(\textbf{A}_2\textbf{H}) ~ \cdots ~ \widetilde{vec}(\textbf{A}_K\textbf{H})] .
\end{equation*}
Consider the vector of transformed information symbols \mbox{$\textbf{s} = {\bf \Theta}^{-1}\textbf{x}$} which takes values from \mbox{$\mathcal{A}'={\bf \Theta}^{-1}\mathcal{A} \subset \mathbb{Z}^K$}, where $\mathbb{Z}$ denotes the ring of integers. The components of $\textbf{s}$ take finite integer values, i.e., $s_i \in \mathcal{A}_q \subset \mathbb{Z}$, with $\vert \mathcal{A}_q \vert =q$ for some finite $q$. Hence one can use a sphere decoder \cite{ViB} to decode $\textbf{s}$ and then obtain the ML estimate of the information vector $\textbf{x}$. The ML decoder output is given by
\begin{equation} \label{eq:ML_decoder}
\check{\textbf{s}} = \operatorname{arg}~\min_{\textbf{s} \in \mathcal{A}'} \left \Vert \textbf{y}-\textbf{G}{\bf \Theta}  \textbf{s}\right \Vert_F^2,
\end{equation}
where $||\cdot||_F$ denotes the Frobenius norm of a matrix.

\subsection{Motivation for our results}

It is claimed in~\cite{HaH} without proof that \mbox{$R \leq min\{M,N\}$} is a sufficient condition for the system of equations defined by \eqref{eq:ML_decoder} to be \emph{not underdetermined}, i.e., for
\begin{equation*}
 rank(\textbf{G})= rank(\textbf{G}{\bf \Theta}) = K, ~~\textrm{with} ~~ K \leq 2MT.
\end{equation*}
In Section \ref{sec2}, we show that the claim made in~\cite{HaH} is true only for  \mbox{$M \geq N$}. This observation is the gateway to the new results presented from Section~\ref{sec2} onwards. 

For a system where \mbox{$rank(\textbf{G})=K$}, the sphere decoder complexity, averaged over noise and channel realizations, is {\it independent of the constellation size $q$} and is roughly polynomial in the dimension of the sphere decoding search \cite{ViB}, \cite{DCB}, \cite{PKY}. However, if the rank of $\textbf{G}$ is less than $K$, the average sphere decoding complexity is no more independent of the constellation size. When $rank(\textbf{G}) = K'<K$, the conventional sphere decoder needs to be modified as follows~\cite{DEC}: The $\textbf{R}$ matrix resulting from the $\textbf{QR}$-decomposition of $\textbf{G}{\bf \Theta}$ has the form \mbox{$\textbf{R}=[\textbf{R}_a~\textbf{R}_b] \in \mathbb{R}^{K' \times K}$}, where $\textbf{R}_a$ is a $K' \times K'$ upper triangular full-rank matrix. There is a corresponding partition of $\textbf{s}$ as $[\textbf{s}_a^T ~\textbf{s}_b^T]^T$. If $q$ is the size of the regular PAM constellation used, then for each of the $q^{K-K'}$ values of $\textbf{s}_b$, the conditionally optimal estimate of $\textbf{s}_a$ can be found by first removing the interference from $\textbf{s}_b$, and then using a sphere decoder with the upper triangular matrix $\textbf{R}_a$ to obtain an estimate of $\textbf{s}_a$. Then, from the resulting $q^{K-K'}$ estimates of $\textbf{s}$, the optimal vector is chosen. This observation leads to the following lemma. 

\begin{lemma}\label{singular_thm}
For an STBC whose equivalent channel matrix $\textbf{G}$ is such that $rank(\textbf{G}) = K'< K$, the sphere decoding complexity as a function of $q$ is of the order of $q^{K-K'}$, i.e., $\mathcal{O}\left(q^{K-K'}\right)$.
\end{lemma}

In the rest of the paper, by sphere decoding complexity we mean the average sphere decoding complexity and the focus is on {\it the dependence of the complexity on the constellation size $q$ and not on the dimension of the sphere decoding search}. We now introduce the notion of \emph{singularity} of an STBC which is a direct indicator of its sphere decoding complexity.
\begin{definition}\label{singular}
Let $\mathcal{C}$ be an STBC with rate $R = \frac{K}{2T}$ and let \mbox{$M \geq R$}. We say that $\mathcal{C}$ is \emph{singular} for $M$ receive antennas if \mbox{$rank(\textbf{G}) < K $} with probability $1$. Otherwise, $\mathcal{C}$ is said to be \emph{non-singular} for $M$ receive antennas. 
\end{definition}

To illustrate the effect of the rank of $\textbf{G}$ and singularity of an STBC on the sphere decoding complexity as a function of the constellation size, consider the following example.

\begin{example} \label{ex:FGD}
In~\cite{RGYS}, a rate $R=\frac{17}{8}$ code for $N=4$ antennas with $K=17$ real symbols was presented. Define $\textbf{X}=\begin{bmatrix} 0 & 1 \\ 1 & 0  \end{bmatrix}$ and $\textbf{Z}=\begin{bmatrix} 0 & 1 \\ 1 & 0  \end{bmatrix}$, and for any two matrices $\textbf{A}$, $\textbf{B}$, denote their Kronecker product  by \mbox{$\textbf{A} \otimes \textbf{B}$}. Then, the $17$ weight matrices \mbox{$\textbf{A}_1,\dots,\textbf{A}_{17}$} of the STBC in~\cite{RGYS} are as follows: 
\begin{equation*}
\begin{array}{lll}
\textbf{A}_1 = \textbf{I}_2 \otimes \textbf{I}_2, & \textbf{A}_2 = i\textbf{Z} \otimes \textbf{I}_2, & \textbf{A}_3 = \textbf{Z}\textbf{X} \otimes \textbf{I}_2, \\
\textbf{A}_4 = i\textbf{X} \otimes \textbf{Z}, & \textbf{A}_5 = \textbf{X} \otimes \textbf{Z}\textbf{X}, & \textbf{A}_6 = i\textbf{X} \otimes \textbf{X}, \\
\textbf{A}_7 = \textbf{I}_2 \otimes \textbf{Z}\textbf{X}, & \textbf{A}_8 = i\textbf{Z} \otimes \textbf{X}, & \textbf{A}_9 = \textbf{Z}\textbf{X} \otimes \textbf{X}, \\
\textbf{A}_{10} = i\textbf{Z}\textbf{X} \otimes \textbf{Z}\textbf{X}, & \textbf{A}_{11} = \textbf{Z}\textbf{X} \otimes \textbf{X}, &\textbf{A}_{12} = i\textbf{X} \otimes \textbf{I}_2, \\
\textbf{A}_{13} = \textbf{Z} \otimes \textbf{Z}\textbf{X}, & \textbf{A}_{14} = i\textbf{I}_2 \otimes \textbf{X}, &\textbf{A}_{15} = i\textbf{Z} \otimes \textbf{Z}, \\
\textbf{A}_{16} = i\textbf{I}_2 \otimes \textbf{I}_2~~and &\textbf{A}_{17} = i\textbf{I}_2 \otimes \textbf{Z}.
\end{array}
\end{equation*}
Let the number of receive antennas be $M=3$. Then, \mbox{$M>R$} and the equivalent channel matrix $\textbf{G}$ is of size \mbox{$24 \times 17$}. Now consider the following randomly generated channel matrix 

{\footnotesize
\begin{equation*}
\textbf{H} = \begin{bmatrix} [r]
   0.3457 + 0.2299i &  0.2078 - 0.0723i & -0.7558 - 0.6116i \\
   0.7316 - 0.5338i & -0.5567 - 0.1707i & -0.5724 - 0.0212i \\
   0.5140 + 0.9689i &  0.6282 + 0.2257i & -2.0819 - 0.1166i \\
  -0.2146 - 1.2102i & -0.8111 + 0.2212i &  1.0171 + 0.4439i
\end{bmatrix}.
\end{equation*}
}

\noindent
The resulting $\textbf{G}$ matrix has rank only $16$ and the structure of the upper triangular matrix $\textbf{R}$ of size \mbox{$24 \times 17$} obtained upon its $\textbf{QR}$-decomposition is shown in~\eqref{eq:R_FGD_code} at the top of next page. The non-zero entries of $\textbf{R}$ are denoted by `$a$'. It is clear that the first $16$ columns of $\textbf{R}$ are linearly independent and the last column lies in the span of the first $16$ columns. Removing the last $8$ rows of $\textbf{R}$, which are all zeros, we get the \mbox{$16 \times 17$} real matrix $\textbf{R}^\prime$ which is used by the sphere decoder to find the ML estimate of the information vector. In this case, from Lemma \ref{singular_thm}, the complexity of sphere decoding this STBC for this particular channel realization $\textbf{H}$ is of the order of \mbox{$q^{17-16}=q$}. In Section~\ref{sec3a}, by establishing that \mbox{$rank(\textbf{G})=16$} with probability $1$, we show that this STBC is singular for $M=3$. 
\begin{figure*}
\begin{equation} \label{eq:R_FGD_code}
\textbf{R} = \begin{bmatrix}
a&0&0&0&0&0&0&0&0&0&0&0&0&0&0&0&0\\
0&a&0&0&0&0&a&a&0&0&0&0&a&a&a&a&a\\
0&0&a&0&0&0&a&0&a&a&a&0&0&a&0&a&a\\
0&0&0&a&0&0&0&a&a&a&0&a&a&0&0&a&a\\
0&0&0&0&a&0&a&a&a&0&a&a&0&0&a&a&0\\
0&0&0&0&0&a&0&0&0&a&a&a&a&a&a&a&0\\
0&0&0&0&0&0&a&a&a&a&a&a&a&a&a&a&a\\
0&0&0&0&0&0&0&a&a&a&a&a&a&a&a&a&a\\
0&0&0&0&0&0&0&0&a&a&a&a&a&a&a&a&a\\
0&0&0&0&0&0&0&0&0&a&a&a&a&a&a&a&a\\
0&0&0&0&0&0&0&0&0&0&a&a&a&a&a&a&a\\
0&0&0&0&0&0&0&0&0&0&0&a&a&a&a&a&a\\
0&0&0&0&0&0&0&0&0&0&0&0&a&a&a&a&a\\
0&0&0&0&0&0&0&0&0&0&0&0&0&a&a&a&a\\
0&0&0&0&0&0&0&0&0&0&0&0&0&0&a&a&a\\
0&0&0&0&0&0&0&0&0&0&0&0&0&0&0&a&a\\
0&0&0&0&0&0&0&0&0&0&0&0&0&0&0&0&0\\ 
0&0&0&0&0&0&0&0&0&0&0&0&0&0&0&0&0\\
0&0&0&0&0&0&0&0&0&0&0&0&0&0&0&0&0\\
0&0&0&0&0&0&0&0&0&0&0&0&0&0&0&0&0\\
0&0&0&0&0&0&0&0&0&0&0&0&0&0&0&0&0\\
0&0&0&0&0&0&0&0&0&0&0&0&0&0&0&0&0\\
0&0&0&0&0&0&0&0&0&0&0&0&0&0&0&0&0\\
0&0&0&0&0&0&0&0&0&0&0&0&0&0&0&0&0\\
\end{bmatrix}
\end{equation}
\hrule
\end{figure*}
\end{example}

\subsection{Contributions of the Paper}
The contributions and organization of this paper are as follows.
\begin{itemize}

\item We introduce the notion of singularity of an STBC which is a direct indicator of its sphere decoding complexity. 

\item We show that contrary to the claim made in~\cite{HaH}, \mbox{$R \leq min\{M,N\}$} is not a sufficient condition for an STBC to be non-singular for $M$ receive antennas (Theorem~\ref{th:main_theorem}, Section~\ref{sec2}). We show that the case of singular STBCs arises only for asymmetric MIMO systems with  $M < N$ (Proposition~\ref{pr:M_less_than_N}, Section~\ref{sec2}).

\item We show that all known families of high rate \mbox{($R>1$)} multigroup ML decodable\footnote{An STBC is $g$-group or multigroup ML decodable if its symbols can be partitioned into $g$ groups and each group of symbols can be ML decoded independent of others.}  codes \cite{RGYS}, \cite{YGT}, \cite{RYGGZ}, \cite{SrR1} and \cite{NaR1} are non-singular for certain values of $M$ (Section~\ref{sec3}, see Table~\ref{tbl:comparison} for a summary of results). 

\item We derive the sphere decoding complexity of almost all known high-rate multigroup ML decodable codes and show that in each case the sphere decoding complexity is a decreasing function of the number of receive antennas $M$ (Section~\ref{sec3}, see Table~\ref{tbl:comparison}). We show that even when an STBC is singular, multigroup ML decodability helps reduce the sphere decoding complexity. The reduction in complexity is from $\mathcal{O}\left( q^{K-K'} \right)$ to $\mathcal{O}\left( q^{\frac{K-K'}{g}} \right)$, where $g$ is the number of ML decoding groups (Section~\ref{sec3}).

\end{itemize}

Some related open problems are discussed in Section~\ref{sec4}.

{\sf \bf Notations:} Throughout the paper, matrices (vectors) are denoted in bold, uppercase (lowercase) letters. For a complex matrix $\textbf{A}$, its transpose, conjugate and conjugate-transpose are denoted by $\textbf{A}^T$, $\bar{\textbf{A}}$ and $\textbf{A}^H$, respectively. For a square matrix $\textbf{A}$, $det(\textbf{A})$ denotes its determinant. The $n \times n$ identity matrix is denoted by $\textbf{I}_n$ and ${\bf 0}$ is the null matrix of appropriate dimension. Unless used as a subscript or to denote indices, $i = \sqrt{-1}$. For square matrices $\textbf{A}_j$, $j=1,\dots,d$, $diag(\textbf{A}_1,\dots,\textbf{A}_d)$ denotes the square, block-diagonal matrix with $\textbf{A}_1,\dots,\textbf{A}_d$ on the diagonal, in that order. The field of complex numbers is denoted by $\mathbb{C}$. 

\section{Basic results on the rank of the equivalent channel matrix} \label{sec2}

We present a few results which we will use in the following sections to derive the rank of $\textbf{G}$ for specific STBCs. The following result shows that if any STBC is singular for $M$ receive antennas, then $M<N$. Thus, the rank-deficiency problem arises only in asymmetric MIMO with \mbox{$M < N$}.
\begin{proposition} \label{pr:M_less_than_N}
If $M \geq N$, every $T \times N$ STBC is non-singular for $M$ receive antennas.
\end{proposition}
\begin{proof}
Since $\widetilde{vec}(\cdot)$ is an isomorphism from the $\mathbb{R}$-vector space $\mathbb{C}^{T \times M}$ to $\mathbb{R}^{2MT}$, it is enough to show that $\textbf{A}_1\textbf{H},\dots,\textbf{A}_K\textbf{H}$ are linearly independent with probability $1$. Suppose $\textbf{H}$ is full-ranked, i.e., $rank(\textbf{H})=N$, then there exists a matrix $\textbf{H}^\dag \in \mathbb{C}^{M \times N}$ such that $\textbf{H}\textbf{H}^\dag = \textbf{I}_N$. If $\textbf{V}=\sum_{i=1}^{K}{a_i\textbf{A}_i\textbf{H}}={\bf 0}$, it would mean that $\textbf{V}\textbf{H}^\dag = \sum_{i=1}^{K}a_i\textbf{A}_i = {\bf 0}$. Since $\textbf{A}_i$ are the weight matrices of an STBC, they are linearly independent and hence $a_i=0$, $i=1,\dots,K$. Thus $rank(\textbf{G})=K$ if $\textbf{H}$ is full-ranked. Since $\textbf{H}$ is full-ranked with probability $1$~\cite{TsV}, we have shown that any STBC is non-singular for $M$ receive antennas if $M \geq N$.
\end{proof}

Let \mbox{$\langle \textbf{A}_1,\textbf{A}_2,\dots,\textbf{A}_K \rangle$} denote the $\mathbb{R}$-linear subspace of $\mathbb{C}^{T \times N}$ spanned by the matrices \mbox{$\textbf{A}_1,\dots,\textbf{A}_K$}. 

\begin{proposition}	
	\label{pr:rank_independent_basis}
Let $\textbf{B}_1,\dots,\textbf{B}_K$ be $T \times N$ complex matrices such that \mbox{$\langle \textbf{A}_1,\dots,\textbf{A}_K \rangle = \langle \textbf{B}_1,\dots,\textbf{B}_K \rangle$} and let $\textbf{H} \in \mathbb{C}^{N \times M}$ be any matrix. Then the column spaces of the matrices
\begin{align}
	&\textbf{G}_\textbf{A}(\textbf{H}) = [\widetilde{vec}(\textbf{A}_1\textbf{H})~\widetilde{vec}(\textbf{A}_2\textbf{H}) ~ \cdots ~ \widetilde{vec}(\textbf{A}_K\textbf{H})] \textrm{ and} \label{eq:GA_of_H}\\
	&\textbf{G}_\textbf{B}(\textbf{H}) = [\widetilde{vec}(\textbf{B}_1\textbf{H})~\widetilde{vec}(\textbf{B}_2\textbf{H}) ~ \cdots ~ \widetilde{vec}(\textbf{B}_K\textbf{H})] \label{eq:GB_of_H}
\end{align}
are identical. In particular, \mbox{$rank(\textbf{G}_\textbf{A}(\textbf{H})) = rank(\textbf{G}_\textbf{B}(\textbf{H}))$}.
\end{proposition}
\begin{proof}
	Let $\textbf{v}$ be any vector in the column space of $\textbf{G}_\textbf{A}(\textbf{H})$. Then, $\textbf{v} = \sum_{i=1}^{K}{a_i\widetilde{vec}(\textbf{A}_i\textbf{H})}$, for some choice of real numbers $a_i$, $i=1,\dots,K$. Since each of the \mbox{$\textbf{A}_i \in \langle \textbf{B}_1,\dots,\textbf{B}_K \rangle$}, every $\textbf{A}_i$ can be written as some real linear combination of $\textbf{B}_1,\dots,\textbf{B}_K$. It follows that every $\widetilde{vec}(\textbf{A}_i\textbf{H})$ can be written as some real linear combination of \mbox{$\widetilde{vec}(\textbf{B}_1\textbf{H}),\dots,\widetilde{vec}(\textbf{B}_K\textbf{H})$}. Hence $\textbf{v}$ belongs to the column space of $\textbf{G}_\textbf{B}(\textbf{H})$. Similarly we can show that every vector in the column space of $\textbf{G}_\textbf{B}(\textbf{H})$ belongs to the column space of $\textbf{G}_\textbf{A}(\textbf{H})$ also. This completes the proof.
\end{proof}

The following result shows that if every weight matrix of a given STBC is multiplied on the left by a constant invertible matrix, then the rank of the equivalent channel matrix is unchanged.

\begin{proposition} \label{pr:rank_independent_C_mult}
	Let $\textbf{C} \in \mathbb{C}^{T \times T}$ be any full-rank matrix, $\textbf{B}_i=\textbf{C}\textbf{A}_i$, \mbox{$i=1,\dots,K$} and $\textbf{H}$ be any $N \times M$ complex matrix. Then we have \mbox{$rank(\textbf{G}_\textbf{A}(\textbf{H})) = rank(\textbf{G}_\textbf{B}(\textbf{H}))$}, where $\textbf{G}_\textbf{A}(\textbf{H})$ and $\textbf{G}_\textbf{B}(\textbf{H})$ are as defined in~\eqref{eq:GA_of_H} and~\eqref{eq:GB_of_H}.
\end{proposition}
\begin{proof}
It suffices to show that the subspaces \mbox{$\langle \textbf{A}_1\textbf{H},\dots,\textbf{A}_K\textbf{H}\rangle$} and \mbox{$\langle \textbf{C}\textbf{A}_1\textbf{H},\dots,\textbf{C}\textbf{A}_K\textbf{H} \rangle$} have the same dimension. The proof is complete if we show that the vector space homomorphism $\varphi$ from the former subspace to the latter, that sends $\textbf{V}=\sum_{i=1}^{K}{a_i\textbf{A}_i\textbf{H}}$ to $\sum_{i=1}^{K}{a_i\textbf{C}\textbf{A}_i\textbf{H}}$ is a one to one map. If $\textbf{V} \in ker(\varphi)$, then \mbox{$\varphi(\textbf{V})=\textbf{C}\sum_{i=1}^{K}{a_i\textbf{A}_i\textbf{H}}=\textbf{C}\textbf{V}={\bf 0}$}. Since $\textbf{C}$ is invertible, this means that $\textbf{V}={\bf 0}$. This completes the proof.
\end{proof}

It is shown in~\cite{NaR1} that for any $N \geq 1$, there exists an explicitly constructable set of $N^2$ matrices belonging to $\mathbb{C}^{N \times N}$ that are unitary, Hermitian and linearly independent over $\mathbb{R}$. This set of matrices forms a basis for the space of $N \times N$ Hermitian matrices. Denote by $\mathcal{C}_N^{Herm}$ any STBC obtained by using these $N^2$ matrices as weight matrices. For positive integers $n$ and $m$, define the function 
\begin{equation*}
f(n,m)=n^2 - ( (n-m)^+ )^2, 
\end{equation*}
where \mbox{$(a)^+=max\{a,0\}$}. We now state the main result of this paper.
\begin{theorem}	\label{th:main_theorem}
The equivalent channel matrix of the STBC $\mathcal{C}_N^{Herm}$ for $M$ receive antennas has rank $f(N,M)$ with probability $1$.
\end{theorem}
\begin{proof}
See Appendix A.
\end{proof}

The rate of $\mathcal{C}_N^{Herm}$ is $\frac{N}{2}$ and the rank of the equivalent channel is less than \mbox{$K=N^2$} whenever \mbox{$M < N$}. Thus, this STBC is singular for all \mbox{$\frac{N}{2} \leq M < N$}.
\begin{example} \label{ex:X_3}
Consider the STBC $\mathcal{C}_3^{Herm}$ used in an asymmetric MIMO system with \mbox{$M=2$} receive antennas. In this case, \mbox{$R = \frac{3}{2} < min\{M,N\}$} and the equivalent channel matrix $\textbf{G}$ is of size \mbox{$12 \times 9$}. From Theorem~\ref{th:main_theorem}, we know that the rank of $\textbf{G}$ is equal to \mbox{$f(3,2)=8$} with probability $1$. Hence, this STBC is singular for $2$ receive antennas. The $9$ weight matrices of the STBC $\mathcal{C}_3^{Herm}$ are as follows
{\small
\begin{align}
\textbf{A}_1=\begin{bmatrix} 1 & 0 & 0 \\ 0 & 1 & 0 \\ 0 & 0 & 1  \end{bmatrix},~\textbf{A}_2=\begin{bmatrix} 1 & 0 & 0 \\ 0 & -1 & 0 \\ 0 & 0 & 1  \end{bmatrix},~\textbf{A}_3=\begin{bmatrix} 1 & 0 & 0 \\ 0 & 1 & 0 \\ 0 & 0 & -1  \end{bmatrix}, \nonumber 
\end{align}
\begin{align}
\textbf{A}_4=\begin{bmatrix} 0 & 1 & 0 \\ 1 & 0 & 0 \\ 0 & 0 & 1  \end{bmatrix},~\textbf{A}_5=\begin{bmatrix} 0 & i & 0 \\ -i & 0 & 0 \\ 0 & 0 & 1  \end{bmatrix},~\textbf{A}_6=\begin{bmatrix} 0 & 0 & 1 \\ 0 & 1 & 0 \\ 1 & 0 & 0  \end{bmatrix}, \nonumber
\end{align}
\begin{align}
\textbf{A}_7=\begin{bmatrix} 0 & 0 & i \\ 0 & 1 & 0 \\ -i & 0 & 0  \end{bmatrix},~\textbf{A}_8=\begin{bmatrix} 1 & 0 & 0 \\ 0 & 0 & 1 \\ 0 & 1 & 0  \end{bmatrix}~\&~\textbf{A}_9=\begin{bmatrix} 1 & 0 & 0 \\ 0 & 0 & i \\ 0 & -i & 0  \end{bmatrix}. \nonumber
\end{align}
}
The structure of the \mbox{$12 \times 9$} upper triangular matrix $R'$ obtained from the \textbf{QR} decomposition of $\textbf{G}$ when $\textbf{H}$ equals the following randomly generated matrix 
\begin{equation*}
\begin{bmatrix}[r]      
   -0.5688 - 0.8117i & -0.1723 + 1.8282i\\
   0.4926 + 0.0742i &  0.1525 - 0.4716i\\
   0.5905 + 0.5107i & -0.8244 + 0.1325i   \end{bmatrix},
\end{equation*}
is given by
\begin{equation*}
R' = \begin{bmatrix}
a&a&a&a&a&a&a&a&a\\
0&a&a&a&a&a&a&a&a\\
0&0&a&a&a&a&a&a&a\\
0&0&0&a&a&a&a&a&a\\
0&0&0&0&a&a&a&a&a\\
0&0&0&0&0&a&a&a&a\\
0&0&0&0&0&0&a&a&a\\
0&0&0&0&0&0&0&a&a\\
0&0&0&0&0&0&0&0&0\\
0&0&0&0&0&0&0&0&0\\
0&0&0&0&0&0&0&0&0\\
0&0&0&0&0&0&0&0&0
\end{bmatrix}.
\end{equation*}

\noindent
The sphere decoder uses the \mbox{$8 \times 9$} matrix $R$ obtained from $R'$ by deleting its last $4$ rows which are all zero. Hence, for this particular channel realization $\textbf{H}$, the sphere decoding complexity is of the order of \mbox{$q^{9-8}=q$}.
\end{example}

The remaining part of this section is concerned with multigroup ML decodable codes. Suppose the code obtained from an STBC $\mathcal{C}$ with a signal set $\mathcal{A}$ is $g$-group ML decodable for some \mbox{$g>1$}. The information symbols $\{x_1,\dots,x_K\}$ can be partitioned into $g$ vectors $\textbf{x}_{\mathcal{I}_1},\dots,\textbf{x}_{\mathcal{I}_g}$ of length $\lambda_1,\dots,\lambda_g$ respectively such that each symbol vector can be ML decoded independently of other symbol vectors. There is a corresponding partition of the channel matrix into $g$ submatrices $\textbf{G}_1,\dots,\textbf{G}_g$, such that 
\begin{equation}\label{submatrices}
  \textbf{Gx} = \sum_{k=1}^{g}{\textbf{G}_k\textbf{x}_{\mathcal{I}_k}},~~\textbf{G}_k \in \mathbb{R}^{2MT \times \lambda_k}, ~~ \textrm{for} ~ k=1,\dots,g.
\end{equation}
 In Theorem 2 of~\cite{SrR2} it is shown that for any $k \neq k'$ and any channel realization $\textbf{H}$, every column of $\textbf{G}_k$ is orthogonal to every column of $\textbf{G}_{k'}$. As a direct consequence of this, we have the following proposition.
\begin{proposition} \label{pr:sum_of_ranks_multigroup}
For any $g$-group ML decodable STBC and any channel realization $\textbf{H}$, \mbox{$rank(\textbf{G}) = \sum_{k=1}^{g}rank(\textbf{G}_k)$}.
\end{proposition}
\begin{proof}
Since the column spaces of $\textbf{G}_k$, $k=1,\dots,g$ are orthogonal to each other, the dimension of the column space of $\textbf{G}$ is equal to the sum of the dimensions of the column spaces of $\textbf{G}_k$, $k=1,\dots,g$.  
\end{proof}

\section{Sphere decoding complexity of some known families of codes} \label{sec3}

In this section, we show that all known families of high-rate \mbox{($R>1$)} multigroup ML decodable codes are singular for certain number of receive antennas. Using the properties of the rank of the equivalent channel matrix derived in the previous section, we now derive the sphere decoding complexities of these known multigroup ML decodable STBCs. Table~\ref{tbl:comparison} summarizes the results of this section. The table lists the sphere decoding complexity and the minimum number of receive antennas for non-singularity of $\mathcal{C}_N^{Herm}$, the codes in~\cite{RGYS},~\cite{RYGGZ},~\cite{SrR1}, and the codes in~\cite{NaR1} corresponding to even number of ML decoding groups.


\renewcommand{\arraystretch}{1.75}
\begin{center}
\begin{table*}[thbp]
\begin{threeparttable}
\caption{Comparison of Sphere Decoding Complexities of known singular STBCs}
\begin{tabular} {|l|c|c|c|c|c|c|}
\hline 
\multirow{2}{*}{Code} &  Transmit          &  Delay & Groups           & Rate         & Minimum $M$      &  Order of Sphere            \\ 
                      &    Antennas        &        &                  & (cspcu)      &    for           &  Decoding Complexity                       \\
                      &    $N$             &   $T$  &  $g$             & $R$          &  non-singularity &  Exponent of $q^{\dag,*}$                   \\
\hline 
$\mathcal{C}_N^{Herm}$ (Theorem~\ref{th:main_theorem})   & $\geq 1$           &   $N$  &   $1$            & $\frac{N}{2}$&    $N$           &   $\left((N-M)^+\right)^2$  \\
\hline
Ren et. al.~\cite{RGYS}& $4$               &   $4$  &  $2$             &$\frac{17}{8}$& $4$              &   $\left((4-M)^+\right)^2$ \\
\hline
Ren et al.~\cite{RYGGZ}&$\geq 1$           &$\geq 2N$&  $2$            &$N-\frac{N^2-1}{T}$& $N$  & $(N-M)^+ \cdot (T-N-M)$  \\
\hline
Srinath et al.~\cite{SrR1}&$2^m$, $m \geq 2$&   $N$  &   $2$           &$\frac{N}{4}+\frac{1}{N}$& $\frac{N}{2}$  &$\left((\frac{N}{2}-M)^+\right)^2$ \\
\hline
\multirow{2}{*}{Natarajan et. al.~\cite{NaR1}$^\ddag$}&${ng2^{\lfloor \frac{g-1}{2} \rfloor}}$, ${n \geq 1}$&${N}$&$2\ell,~\ell \geq 1$&${\frac{N}{g2^{g-1}}+\frac{g^2-g}{2N}}$& $\frac{N}{g2^{g-2}}$  &$\left(\left(\frac{N}{g2^{\lfloor \frac{g-1}{2} \rfloor}} - 2^{\lfloor \frac{g-1}{2} \rfloor}M\right)^+\right)^2$ \\
\cline{2-7}
&${n2^{\lfloor \frac{g-1}{2} \rfloor}}$, ${n \geq 1}$ &${gN}$ & $2\ell,~\ell \geq 1$ &${\frac{N}{2^{g-1}}+\frac{g-1}{2N}}$& $\frac{N}{2^{g-2}}$ & $\left(\left(\frac{N}{2^{\lfloor \frac{g-1}{2} \rfloor}} - 2^{\lfloor \frac{g-1}{2} \rfloor}M\right)^+\right)^2$ \\
\hline                                  
\end{tabular}
\label{tbl:comparison}

  \begin{tablenotes}
    \item[$\dag$] The size of the real constellation used is denoted by $q$. 
    \item[$*$] For any real number $a$, $a^+$ is defined as $max\{a,0\}$.
    \item[$\ddag$] \cite{NaR1} contains codes for all \mbox{$g>1$} and not just even values of $g$.
  \end{tablenotes}
\end{threeparttable}
\hrule
\end{table*}
\end{center}
\renewcommand{\arraystretch}{1}


\subsection{Fast-group-decodable STBC from Ren et. al.~\cite{RGYS}} \label{sec3a}

In \cite{RGYS}, a 2-group decodable STBC for 4 transmit antennas with $R=17/8$ was constructed. For this code, with the notations as used in \eqref{submatrices}, $\textbf{x}_{\mathcal{I}_1} = x_1$, $\lambda_1 = 1$, $\textbf{x}_{\mathcal{I}_2} = [ x_2,\dots,x_{16}]^T$, $\lambda_2 = 16$ and $\textbf{G} = [ \textbf{G}_1 ~~ \textbf{G}_2]$. The weight matrix corresponding to $x_1$ is $\textbf{I}_4$. Since any two weight matrices from different groups are Hurwitz-Radon orthogonal, i.e., satisfy \mbox{$\textbf{A}_i^H\textbf{A}_j + \textbf{A}_j^H\textbf{A}_i = {\bf 0}$}, all the matrices in the second group must be skew-Hermitian, i.e., $\textbf{A}_i^H = - \textbf{A}_i$, $i = 2,\cdots,17$. As a result, we can express these weight matrices as $\textbf{A}_i = i\textbf{I}_4.\textbf{A}_i^\prime$, $i = 2,\cdots,17$, with $\textbf{A}_i^\prime$ being Hermitian matrices. It is clear that for any arbitrary channel realization $\textbf{H}$, $\textbf{G}_1=[\widetilde{vec}(\textbf{H})]$ which is non-zero with probability $1$. Hence, $rank(\textbf{G}_1) = 1$ with probability 1. From Proposition~\ref{pr:rank_independent_C_mult} and Theorem~\ref{th:main_theorem}, $rank(\textbf{G}_2) = f(4,M)$ with probability $1$. So,
\begin{enumerate}
 \item for $M=3$, $rank(\textbf{G}_2) = 15$. Since $\textbf{x}_{\mathcal{I}_1}$ and $\textbf{x}_{\mathcal{I}_2}$ can be decoded independently of each other, from Proposition \ref{singular}, the sphere decoding complexity of the first group is independent of $q$ while that of the second group is $\mathcal{O}(q^{(16-15)})=\mathcal{O}(q)$. Consequently, the sphere decoding complexity of the code in~\cite{RGYS} is $\mathcal{O}(q)$ for 3 receive antennas.
\item For $M\geq 4$, $f(4,M)=16$ and the STBC in~\cite{RGYS} is non-singular for 4 or more receive antennas. Hence, its sphere decoding complexity is independent of $q$. 
\end{enumerate}

\subsection{High-rate $2$-group ML decodable codes from Srinath et. al.~\cite{SrR1}} \label{sec3b}

A family of $2$-group ML decodable STBCs was constructed in~\cite{SrR1} for $N=2^m$, $m>1$ antennas with rate $R=\frac{N}{4}+\frac{1}{N}$ cspcu. This family includes the rate $\frac{5}{4}$ code of~\cite{YGT} for $N =4$ as a special case. The number of symbols in the STBC is \mbox{$K=\frac{N^2}{2}+2$}. In the rest of this subsection we show that the sphere decoding complexity is $\mathcal{O}\left( q^{\left(\left(\frac{N}{2}-M\right)^+\right)^2} \right)$ which is polynomial in $q$, and so is large for all \mbox{$\lceil R \rceil \leq M < \frac{N}{2}$}, where $\lceil a \rceil$ is the smallest integer greater than or equal to $a$. Note that the sphere decoding complexity is a decreasing function of the number of receive antennas $M$.

\subsubsection*{Derivation of sphere decoding complexity}
The STBCs constructed in~\cite{SrR1} have a block diagonal structure. The weight matrices for $N=2^m$ antennas are of the form $diag(\textbf{V}_1,\textbf{V}_2)$, where $\textbf{V}_1,\textbf{V}_2 \in \mathbb{C}^{n \times n}$ are unitary and \mbox{$n=\frac{N}{2}$}. Noting that the STBC is 2-group decodable, denote the set of weight matrices belonging to the first and the second group by $\mathcal{G}_1$ and $\mathcal{G}_2$, respectively. For all the matrices in $\mathcal{G}_1$, $\textbf{V}_1$ is constant, say $\textbf{F}_1$, and for all the matrices in $\mathcal{G}_2$, $\textbf{V}_2$ is constant, say $\textbf{F}_2$. Each group contains \mbox{$n^2+1$} real symbols. We will now derive the rank of the submatrix $\textbf{G}_1$ of $\textbf{G}$ that corresponds to $\mathcal{G}_1$.
 
Let $\mathcal{G}_1 = \{\textbf{A}_1,\dots,\textbf{A}_{n^2+1}\}$. Consider the set 
\begin{equation*}
 \mathcal{G}_1^\prime = \{\textbf{A}_i'=\textbf{C}\textbf{A}_i, i=1,\dots,n^2+1\},
\end{equation*}
where $\textbf{C}=diag(\textbf{I}_n , i\textbf{F}_2^H)$. Since $\textbf{C}$ is unitary, from Proposition~\ref{pr:rank_independent_C_mult}, the rank of $\textbf{G}_1$ equals the rank of $\textbf{G}_1'$, the equivalent channel matrix corresponding to $\mathcal{G}_1^\prime$. Since the multiplication of all the weight matrices of an STBC by a unitary matrix does not affect its multigroup ML decodability, for any matrix $\textbf{B}$ belonging to $\mathcal{G}_2$ and any \mbox{$j=1,\dots,n^2+1$}, we have 
	\begin{equation} \label{eq:sec3a_HR_orth}
		\textbf{A}_j'^H(\textbf{C}\textbf{B}) + (\textbf{C}\textbf{B})^H\textbf{A}_j'={\bf 0},
	\end{equation}
where $\textbf{A}_j'$ and $\textbf{C}\textbf{B}$ are block diagonal and are of the form 
\begin{align}
	\textbf{A}_j' = \begin{bmatrix} \textbf{F}_1 & {\bf 0} \\ {\bf 0} & \textbf{D}_j \end{bmatrix} \textrm{ and } 
\textbf{C}\textbf{B} = \begin{bmatrix} \textbf{V} & {\bf 0} \\ {\bf 0}  & i\textbf{I}_n \end{bmatrix}, \nonumber
\end{align}
for some unitary matrices $\textbf{D}_j$ and $\textbf{V}$. Since $\textbf{A}_j'$ and $\textbf{C}\textbf{B}$ satisfy~\eqref{eq:sec3a_HR_orth} and are block diagonal, from Lemma 1 of~\cite{SrR1}, \mbox{$\textbf{D}_j^Hi\textbf{I}_n + (i\textbf{I}_n)^H\textbf{D}_j={\bf 0}$}. Thus, for \mbox{$j=1,\dots,n^2+1$}, $\textbf{D}_j$ is a \mbox{$n \times n$} unitary, Hermitian matrix. Next we use Proposition~\ref{pr:rank_independent_basis} to find the rank of $\textbf{G}_1'$.

Note that \mbox{$\langle \textbf{A}_1',\dots,\textbf{A}_{n^2+1}' \rangle$} is same as the span of
\begin{equation} \label{eq:F_and_D_direct_sum_form}
	\begin{bmatrix} \textbf{F}_1 & {\bf 0} \\ {\bf 0} & {\bf 0} \end{bmatrix}, \begin{bmatrix}{\bf 0} & {\bf 0} \\ {\bf 0} & \textbf{D}_1 \end{bmatrix},\dots,\begin{bmatrix}[l] {\bf 0} & {\bf 0} \\ {\bf 0} & \textbf{D}_{n^2+1} \end{bmatrix}.
\end{equation}
Since $\textbf{A}_1,\dots,\textbf{A}_{n^2+1}$ are linearly independent, $\textbf{A}_1',\dots,\textbf{A}_{n^2+1}'$ are also linearly independent. Further, the first matrix in~\eqref{eq:F_and_D_direct_sum_form} is linearly independent of the remaining matrices and hence the dimension of the span of the remaining matrices in~\eqref{eq:F_and_D_direct_sum_form} is $n^2$. Without loss of generality, let us assume that $\textbf{D}_1,\dots,\textbf{D}_{n^2}$ are linearly independent, thus $\langle \textbf{D}_1,\dots,\textbf{D}_{n^2} \rangle$ is the space of all $n \times n$ Hermitian matrices. Then, \mbox{$\langle \textbf{A}_1',\dots,\textbf{A}_{n^2+1}' \rangle$} equals the space spanned by
\begin{equation} \label{eq:F_and_direct_sum_reduced} 
	\begin{bmatrix} \textbf{F}_1 & {\bf 0} \\ {\bf 0} & {\bf 0} \end{bmatrix}, \begin{bmatrix}{\bf 0} & {\bf 0} \\ {\bf 0} & \textbf{D}_1 \end{bmatrix},\dots,\begin{bmatrix}[l] {\bf 0} & {\bf 0} \\ {\bf 0} & \textbf{D}_{n^2} \end{bmatrix}.
\end{equation} 
From Proposition~\ref{pr:rank_independent_basis}, it is enough if we concentrate on the STBC whose weight matrices are given by~\eqref{eq:F_and_direct_sum_reduced}. Let the channel matrix be partitioned as \mbox{$\textbf{H} = \begin{bmatrix} \textbf{H}_1 \\ \textbf{H}_2 \end{bmatrix}$}, where \mbox{$\textbf{H}_1,\textbf{H}_2 \in \mathbb{C}^{n \times M}$}. We need to compute the dimension of the space spanned by the weight matrices multiplied on the right by $\textbf{H}$ which is
	\begin{equation*}
		\left\langle \begin{bmatrix}\textbf{F}_1\textbf{H}_1 \\ {\bf 0} \end{bmatrix},\begin{bmatrix} {\bf 0} \\ \textbf{D}_1\textbf{H}_2 \end{bmatrix},\dots,\begin{bmatrix}[c]{\bf 0} \\ \textbf{D}_{n^2}\textbf{H}_2 \end{bmatrix} \right\rangle.
	\end{equation*}
	With probability $1$, $\textbf{H}_1$ is non-zero and hence the first matrix is linearly independent of the remaining matrices. From Theorem~\ref{th:main_theorem}, the dimension of the span of the remaining $n^2$ matrices is \mbox{$f(n,M)=n^2 - ( (n-M)^+ )^2$} with probability $1$. Thus $rank(\textbf{G}_1)$ equals $f(n,M) + 1$ with probability $1$. A similar result can also proved for the second ML decoding group, i.e., for $rank(\textbf{G}_2)$. From Proposition~\ref{pr:sum_of_ranks_multigroup}, \mbox{$rank(\textbf{G}) = rank(\textbf{G}_1) + rank(\textbf{G}_2)$} which equals 
\begin{equation*} 
	K'=2\left( \frac{N^2}{4} - \left(\left(\frac{N}{2} - M\right)^+\right)^2 + 1\right).
\end{equation*}
Comparing this with \mbox{$K=2\left( \frac{N^2}{4} + 1\right)$}, we see that the STBC given in~\cite{SrR1} is non-singular only if \mbox{$M \geq \frac{N}{2}\approx 2R$}. Hence the code is singular for all \mbox{$\lceil R \rceil \leq M < \frac{N}{2}$}.

Now, the two groups of symbols can be ML decoded independently of each other, and the number of symbols in each group is $\frac{K}{2}$, with $rank(\textbf{G}_1) = rank(\textbf{G}_1)= \frac{K'}{2}$. Thus, the sphere decoding complexity of the STBC is $\mathcal{O}\left( q^{\frac{K-K'}{2}} \right)$ instead of $\mathcal{O}\left( q^{K-K'} \right)$. Hence, \emph{multigroup ML decodability reduces the sphere decoding complexity even if the STBC is singular}.


\subsection{Two group ML decodable codes from Ren et. al.~\cite{RYGGZ}}

In~\cite{RYGGZ}, $2$-group ML decodable codes for all $N\geq 1$ and all even $T \geq 2N$ were constructed with rate \mbox{$R = N - \frac{N^2-1}{T}$}. The number of symbols per group is \mbox{$\frac{K}{2}=TN-N^2+1$}. In this subsection, we show that the codes of this family are singular for all \mbox{$\lceil R \rceil \leq M < N$} receive antennas and that their sphere decoding complexity is $\mathcal{O}\left( q^{(N-M)^+ \cdot (T-N-M)} \right)$. Using Proposition~\ref{pr:M_less_than_N}, it is clear that these codes are non-singular if and only if \mbox{$M \geq N$}.

\subsubsection*{Derivation of sphere decoding complexity}

The structure and derivation of the sphere decoding complexity of these codes is similar to that of the codes from~\cite{SrR1}, which was discussed in Section~\ref{sec3b}. The weight matrices of the STBCs in~\cite{RYGGZ} are of the form $\begin{bmatrix} \textbf{V}_1 \\ \textbf{V}_2 \end{bmatrix}$, where $\textbf{V}_1,\textbf{V}_2 \in \mathbb{C}^{\frac{T}{2} \times N}$. For all the matrices in the first group, $\textbf{V}_1= \pm \textbf{F}_1$ for some constant matrix $\textbf{F}_1$, and for all the matrices in the second group, $\textbf{V}_2 = \pm \textbf{F}_2$ for some constant matrix $\textbf{F}_2$. The STBCs constructed in~\cite{RYGGZ} are such that for each $k=1,2$, the $\textbf{V}_k$ submatrices of any two weight matrices belonging to different groups are Hurwitz-Radon orthogonal. We consider the case where $\textbf{F}_k$, $k=1,2$ are semi-unitary i.e., \mbox{$\textbf{F}_k^H\textbf{F}_k = \textbf{I}_N$}. We derive the sphere decoding complexity only for the first group. Using a similar argument the complexity for the second group can be derived, and it is same as that of the first group.

Since $\textbf{F}_2$ is semi-unitary, there exists a unitary \mbox{$T \times T$} matrix $\tilde{\textbf{F}_2}$ such that \mbox{$\tilde{\textbf{F}_2}^H\textbf{F}_2 = \begin{bmatrix} \textbf{I}_N \\ {\bf 0} \end{bmatrix}$}. Consider the new STBC $\mathcal{C}'$ obtained by multiplying all the weight matrices of the original STBC on the left by \mbox{$\textbf{C}=\begin{bmatrix} \textbf{I}_T & {\bf 0} \\ {\bf 0} & i\tilde{\textbf{F}_2}^H \end{bmatrix}$}. Then, the lower submatrix of all the weight matrices of the second group are of the form $\begin{bmatrix} \pm \textbf{I}_N \\ {\bf 0} \end{bmatrix}$. Since the lower submatrix of every matrix in the first group is Hurwitz-Radon orthogonal to $\begin{bmatrix} \pm \textbf{I}_N \\ {\bf 0} \end{bmatrix}$, the weight matrices in the first group of $\mathcal{C}'$ have the following structure: $\begin{bmatrix}[l] \textbf{V}_1 \\ \textbf{B} \\ \textbf{E} \end{bmatrix}$, where \mbox{$\textbf{V}_1= \pm \textbf{F}_1$}, $\textbf{B}$ is an $N \times N$ Hermitian matrix and \mbox{$\textbf{E} \in \mathbb{C}^{\frac{T}{2}-N \times N}$}. Let \mbox{$\textbf{B}_1,\dots,\textbf{B}_{N^2}$} be any basis for the space of $N \times N$ Hermitian matrices over $\mathbb{R}$, $L=TN-2N^2$ and \mbox{$\textbf{E}_1,\dots,\textbf{E}_L$} be 
\begin{align}
			\begin{bmatrix} 1 & 0 & \cdots & 0 \\  \vdots & & & \vdots \\ 0 & 0 & \cdots& 0 \end{bmatrix},\begin{bmatrix} 0 & 1 & \cdots & 0 \\  \vdots & & & \vdots \\ 0 & 0 & \cdots& 0 \end{bmatrix}, \cdots,\begin{bmatrix} 0 & 0 & \cdots & 0 \\  \vdots & & & \vdots \\ 0 & 0 & \cdots& 1 \end{bmatrix}, \nonumber
			\end{align}
			\begin{align}
			\begin{bmatrix} i & 0 & \cdots & 0 \\  \vdots & & & \vdots \\ 0 & 0 & \cdots& 0 \end{bmatrix},\begin{bmatrix} 0 & i & \cdots & 0 \\  \vdots & & & \vdots \\ 0 & 0 & \cdots& 0 \end{bmatrix}, \cdots,\begin{bmatrix} 0 & 0 & \cdots & 0 \\  \vdots & & & \vdots \\ 0 & 0 & \cdots& i \end{bmatrix}, \nonumber 
		\end{align}
which is a standard basis for the space of \mbox{$\left( \frac{T}{2}-N \right) \times N$} complex matrices over $\mathbb{R}$. Clearly the space spanned by the weight matrices of the first group of $\mathcal{C}'$ is a subspace of the space spanned by the following $\frac{K}{2}$ linearly independent matrices:
		\begin{align} \label{eq:unbalanced_code_basis}
			\begin{bmatrix} \textbf{F}_1 \\ {\bf 0} \\ {\bf 0} \end{bmatrix}, \begin{bmatrix} {\bf 0} \\ \textbf{B}_n \\ {\bf 0} \end{bmatrix},~\textrm{for}~n=1,\dots,N^2, \begin{bmatrix} {\bf 0} \\ {\bf 0} \\ \textbf{E}_l \end{bmatrix},~\textrm{for}~l=1,\dots,L. 
		\end{align}

From dimension count, it is clear that the matrices in~\eqref{eq:unbalanced_code_basis} form a basis for the space spanned by the weight matrices of the first group of $\mathcal{C}'$. For any non-zero channel realization \mbox{$\textbf{H} \in \mathbb{C}^{N \times M}$}, the subspaces \mbox{$\mathcal{V}_1=\left\langle \begin{bmatrix} \textbf{F}_1\textbf{H} \\ {\bf 0} \\ {\bf 0} \end{bmatrix} \right\rangle$}, \mbox{$\mathcal{V}_2 = \left\langle \begin{bmatrix} {\bf 0} \\ \textbf{B}_n\textbf{H} \\ {\bf 0} \end{bmatrix},~n=1,\dots,N^2\right\rangle$} and \mbox{$\mathcal{V}_3=\left\langle \begin{bmatrix} {\bf 0} \\ {\bf 0} \\ \textbf{E}_l\textbf{H} \end{bmatrix},~l=1,\dots,L \right\rangle$} are such that their pairwise intersections contain only the all zero matrix. Thus, the rank of the equivalent channel matrix of the first group of $\mathcal{C}'$ is 
			\begin{equation*}
				rank(\textbf{G}_1') = dim(\mathcal{V}_1)+ dim(\mathcal{V}_2)+dim(\mathcal{V}_3).
			\end{equation*}
With probability $1$, \mbox{$dim(\mathcal{V}_1)=1$} and from Theorem~\ref{th:main_theorem}, \mbox{$dim(\mathcal{V}_2) = f(N,M)$}. It is straightforward to show that $dim(\mathcal{V}_3)=(T-2N) \cdot min\{N,M\}$ with probability $1$. From Proposition~\ref{pr:M_less_than_N}, we know that every STBC is non-singular for \mbox{$M \geq N$}. We thus consider only the case \mbox{$M<N$}. Then from Proposition~\ref{pr:rank_independent_C_mult}, the rank of the equivalent channel matrix of the first ML decoding group of the STBC given in~\cite{RYGGZ} is 
		\begin{equation*}
			rank(\textbf{G}_1)=rank(\textbf{G}_1') = (T-2N) \cdot M+ f(N,M) + 1.
		\end{equation*}
Compare this with the number of symbols in the first group \mbox{$\frac{K}{2} = (T-2N) \cdot N + N^2 + 1$}. Thus, for any $M<N$, $rank(\textbf{G}_1) < \frac{K}{2}$ with probability $1$. 

Using a similar argument it can be shown that the rank of the equivalent channel matrix of the second group also equals \mbox{$(T-2N) \cdot M+ f(N,M) + 1$} with probability $1$. Since, the two groups can be ML decoded independently of each other, the complexity of sphere decoding the STBC proposed in~\cite{RYGGZ} is $\mathcal{O}(q^{(N-M) \cdot (T-N-M)})$ for any $M<N$ and the STBC is singular for all \mbox{$\lceil R \rceil \leq M < N$}.

\begin{example}
Consider the STBC from~\cite{RYGGZ} for \mbox{$N=4$} and \mbox{$T=2N=8$}. The rate of this code is \mbox{$R=\frac{17}{8}$} and the number of symbols per decoding group is \mbox{$\frac{K}{2}=17$}. From the above discussion it is clear that this STBC is singular for \mbox{$M=3$}. The rank of the equivalent channel matrix of each ML decoding group equals $16$ with probability $1$ and hence the sphere decoding complexity is $\mathcal{O}(q)$.

It is interesting to compare this code with the code from~\cite{RGYS} which we have discussed in Section~\ref{sec3a} and Example~\ref{ex:FGD}. Both codes have the same parameters $N$, $R$ and both have a sphere decoding complexity that is linear in the constellation size $q$. However, the code in~\cite{RGYS} is fast-group-decodable and $5$ levels can be removed from the sphere decoding search tree of the second decoding group. Hence, after conditioning on the value taken by one of the real symbols (to account for the reduced rank of the equivalent channel matrix), the code in~\cite{RGYS} uses a $10$-dimensional sphere decoder to decode the second group of $16$ symbols. For decoding each ML decoding group, the code from~\cite{RYGGZ}, uses a $16$ dimensional search tree after conditioning on the value of one of the real symbols. Thus, the sphere decoding complexity of the code from~\cite{RGYS} is less than that of~\cite{RYGGZ}.
\end{example}

\subsection{Multigroup ML decodable codes from Natarajan et. al~\cite{NaR1}}

In~\cite{NaR1}, delay optimal $g$-group ML decodable codes were constructed for all \mbox{$g>1$}, \mbox{$N=ng2^{\lfloor \frac{g-1}{2} \rfloor}$}, \mbox{$n \geq 1$}, with rate \mbox{$R=\frac{N}{g2^{g-1}} + \frac{g^2-g}{2N}$}. In this subsection we show that the sphere decoding complexity of the codes with even $g$ is of the order of $q^{\left( \left( n - 2^{\lfloor \frac{g-1}{2}\rfloor} M \right)^+\right)^2}$ and that the codes are non-singular only for \mbox{$M \geq \frac{N}{g2^{g-2}}$}. Also in~\cite{NaR1}, non-delay optimal codes with $T=gN$, $N=n2^{\lfloor \frac{g-1}{2} \rfloor}$, \mbox{$n \geq 1$} were constructed. We show that the sphere decoding complexity of these codes for even values of $g$ is of the order of $q^{\left( \left( n - 2^{\lfloor \frac{g-1}{2}\rfloor} M \right)^+\right)^2}$. Simulation results show that the STBCs in~\cite{NaR1} for odd values of $g$ are also singular for certain number of receive antennas.

\subsubsection*{Delay-optimal codes}

The number of symbols per group is \mbox{$\frac{K}{g} = n^2+g-1$} and let \mbox{$m=2^{\lfloor \frac{g-1}{2} \rfloor}$}. We will now derive the rank of the equivalent channel matrix $\textbf{G}_1$ of the STBC corresponding to the first group. The weight matrices of the first group have a block diagonal structure $diag\left( \textbf{D}_1, \textbf{D}_2, \dots, \textbf{D}_g \right)$, where each \mbox{$\textbf{D}_j \in \mathbb{C}^{nm \times nm}$}. The first block $\textbf{D}_1$ is one of the $n^2$ matrices of the form $\textbf{V} \otimes \textbf{U}_1$, where \mbox{$\textbf{V} \in \mathbb{C}^{n \times n}$} is Hermitian, and the remaining $g-1$ matrices $\textbf{D}_j$, \mbox{$j=2,\cdots,g$} are of the form \mbox{$\pm \textbf{I}_n \otimes \textbf{U}_j$} for some set of $g$ unitary \mbox{$m \times m$} matrices \mbox{$\textbf{U}_1,\dots,\textbf{U}_g$}. Let us multiply all the weight matrices of the first group on the right by \mbox{$\textbf{C}=diag\left(\textbf{I}_n \otimes \textbf{U}_1^H,\cdots,\textbf{I}_n \otimes \textbf{U}_g^H \right)$}. Clearly, the new set of weight matrices \mbox{$\textbf{A}_1',\dots,\textbf{A}_{\frac{K}{g}'}$}  also have a block diagonal structure \mbox{$diag\left( \textbf{D}_1, \textbf{D}_2, \dots, \textbf{D}_g \right)$} where $\textbf{D}_1$ is one of the $n^2$ matrices of the form \mbox{$\textbf{V} \otimes \textbf{I}_n$}, where $\textbf{V}$ is \mbox{$n \times n$} Hermitian and the remaining $g-1$  blocks are of the form \mbox{$\pm \textbf{I}_n \otimes \textbf{I}_m$}. It is straightforward to show that there exists a permutation matrix $\textbf{P}$ such that \mbox{$\textbf{P} \cdot \left( \textbf{V} \otimes \textbf{I}_m \right) \cdot \textbf{P}^T = \textbf{I}_m \otimes \textbf{V}$} for any \mbox{$\textbf{V} \in \mathbb{C}^{n \times n}$}. Consider the matrices \mbox{$\textbf{A}_j'' = \textbf{C}'\textbf{A}_j'\textbf{C}'^T$}, $j=1,\dots,\frac{K}{g}$, where \mbox{$\textbf{C}'=diag\left(\textbf{P},\textbf{I}_{nm},\dots,\textbf{I}_{nm}\right)$}. Let $\textbf{B}_1,\dots,\textbf{B}_{n^2}$ be any basis for the space of \mbox{$n \times n$} Hermitian matrices. From dimension count and the structure of the $\textbf{A}_j''$ matrices, $\langle \textbf{A}_1'',\dots,\textbf{A}_\frac{K}{g}'' \rangle$ equals the span of the matrices $\hat{\textbf{A}}_1,\dots,\hat{\textbf{A}}_\frac{K}{g}$ which given by $diag(\textbf{I}_m \otimes \textbf{B}_l, {\bf 0},\dots,{\bf 0})$, \mbox{$l=1,\dots,n^2$}, $diag({\bf 0}, \textbf{I}_{nm},{\bf 0},\dots,{\bf 0})$, \dots, $diag({\bf 0},\dots,{\bf 0},\textbf{I}_{nm})$. Since $\textbf{C}'^T$ is unitary, the statistics of $\textbf{H}$ and $\textbf{C}'^T\textbf{H}$ are same. Along with Propositions~\ref{pr:rank_independent_basis} and~\ref{pr:rank_independent_C_mult}, it is thus clear that $rank(\textbf{G}_1)$ has the same statistics as the rank of the equivalent channel matrix $\hat{\textbf{G}}_1$ of the STBC whose weight matrices are $\hat{\textbf{A}}_1,\dots,\hat{\textbf{A}}_\frac{K}{g}$.

Let the channel realization be \mbox{$\textbf{H} = \begin{bmatrix} \textbf{H}_1^T & \cdots & \textbf{H}_g^T \end{bmatrix}^T$}, where \mbox{$\textbf{H}_k \in \mathbb{C}^{nm \times M}$}, $k=1,\dots,g$. We are interested in the dimension of \mbox{$\langle \hat{\textbf{A}}_1\textbf{H},\dots,\hat{\textbf{A}}_\frac{K}{g}\textbf{H} \rangle$} which is the span of
	\begin{align} \label{eq:AG_delay_optimal_span}
		\begin{bmatrix} \left(\textbf{I}_m \otimes \textbf{B}_l \right)\textbf{H}_1 \\ {\bf 0} \\ \vdots \\ {\bf 0} \end{bmatrix}, ~\textrm{for}~l=1,\dots,n^2,\begin{bmatrix} {\bf 0} \\ \textbf{H}_2 \\ \vdots \\ {\bf 0} \end{bmatrix}, \dots, \begin{bmatrix} {\bf 0} \\ {\bf 0} \\ \vdots \\ \textbf{H}_g \end{bmatrix}.
	\end{align}
	Thus, with probability $1$, $rank(\hat{\textbf{G}}_1)$ equals the sum of $g-1$ and the dimension of the span of the first $n^2$ matrices in~\eqref{eq:AG_delay_optimal_span}. Let us rewrite $\textbf{H}_1$ as $\begin{bmatrix} \textbf{H}_{1,1}^T & \dots & \textbf{H}_{1,m}^T \end{bmatrix}^T$, where $\textbf{H}_{1,j} \in \mathbb{C}^{n \times M}$ for $j=1,\dots,m$. Thus, the dimension of the span of first $n^2$ matrices in~\eqref{eq:AG_delay_optimal_span} is same as that of
		\begin{align*}
			\begin{bmatrix} \textbf{B}_1\textbf{H}_{1,1} \\ \textbf{B}_1\textbf{H}_{1,2} \\ \vdots \\\textbf{B}_1\textbf{H}_{1,m}  \end{bmatrix}, \begin{bmatrix} \textbf{B}_2\textbf{H}_{1,1} \\ \textbf{B}_2\textbf{H}_{1,2} \\ \vdots \\\textbf{B}_2\textbf{H}_{1,m}  \end{bmatrix}, \dots, \begin{bmatrix} \textbf{B}_{n^2}\textbf{H}_{1,1} \\ \textbf{B}_{n^2}\textbf{H}_{1,2} \\ \vdots \\\textbf{B}_{n^2}\textbf{H}_{1,m}  \end{bmatrix}.
		\end{align*}
This in turn, is equal to the dimension of the span of the following matrices
\begin{align}
&	\begin{bmatrix} \textbf{B}_1\textbf{H}_{1,1} ~  \textbf{B}_1\textbf{H}_{1,2} ~ \cdots ~ \textbf{B}_1\textbf{H}_{1,m}  \end{bmatrix}, \nonumber \\
&	\begin{bmatrix} \textbf{B}_2\textbf{H}_{1,1} ~  \textbf{B}_2\textbf{H}_{1,2} ~ \cdots ~ \textbf{B}_2\textbf{H}_{1,m}  \end{bmatrix}, \nonumber \\
&         ~ ~ ~ ~ ~ ~ ~ ~ ~ ~ ~ ~ ~ ~ ~ ~ ~ ~ ~ ~ \vdots \nonumber \\
&	\begin{bmatrix} \textbf{B}_{n^2}\textbf{H}_{1,1} ~  \textbf{B}_{n^2}\textbf{H}_{1,2} ~ \cdots ~ \textbf{B}_{n^2}\textbf{H}_{1,m}  \end{bmatrix}. \nonumber
\end{align}
From Theorem~\ref{th:main_theorem}, this dimension equals $f(n,mM)$ with probability $1$. Hence, 
\begin{equation*}
	rank(\textbf{G}_1) = rank(\hat{\textbf{G}}_1) = f(n,mM) + g - 1
\end{equation*}
with probability $1$. Compare this with the number of symbols per group \mbox{$\frac{K}{g} = n^2 + g - 1$}. Thus, the delay optimal STBCs of~\cite{NaR1} for even values of $g$ are singular whenever \mbox{$\lceil R \rceil \leq M < \frac{n}{2^{\lfloor \frac{g-1}{2} \rfloor}}$}. Similar results on the rank of the equivalent channel matrix can be proved for other ML decoding groups also. Thus, the sphere decoding complexity of these codes is of the order of $q^{\left( \left( n - 2^{\lfloor \frac{g-1}{2}\rfloor} M \right)^+\right)^2}$. This STBC is singular only for \mbox{$M \geq \frac{N}{g2^{g-2}} \approx 2R$}. For \mbox{$g=2$}, and equal values of $N$, the delay-optimal codes of~\cite{NaR1} and the codes of~\cite{SrR1} have equal rate and the same order of sphere decoding complexity.

\begin{example}
Consider the $g=2$, $N=6$ delay-optimal code of~\cite{NaR1}. The underlying STBC has rate $R=\frac{5}{3}$ and there are $\frac{K}{2}=10$ symbols per ML decoding group. From the ongoing discussion, for \mbox{$M=2$} receive antennas \mbox{$rank(\textbf{G}_k)=9$} with probability $1$ for $k=1,2$. Hence, this STBC is singular for \mbox{$M=2$} and the sphere decoding complexity is $\mathcal{O}(q)$. Let us denote the weight matrices of $\mathcal{C}_3^{Herm}$ given in Example~\ref{ex:X_3} by \mbox{$\textbf{A}_1',\dots,\textbf{A}_9'$}. Then, the $20$ weight matrices of the $N=6$, $g=2$ code of~\cite{NaR1} are given by
\begin{align*}
&\textbf{A}_{\ell} = \begin{bmatrix} i\textbf{A}_{\ell}' & {\bf 0} \\ {\bf 0} & \textbf{I}_3 \end{bmatrix}, \ell =1,\dots,9, ~~\textbf{A}_{10}=\begin{bmatrix} i\textbf{A}_{1}' & {\bf 0} \\ {\bf 0} & -\textbf{I}_3 \end{bmatrix}, \\ \nonumber
&\textbf{A}_{\ell+10} = \begin{bmatrix} \textbf{I}_3 & {\bf 0} \\ {\bf 0} & i\textbf{A}_{\ell}' \end{bmatrix}, \ell =1,\dots,9, ~\textrm{and}~\textbf{A}_{20}=\begin{bmatrix} -\textbf{I}_3 & {\bf 0} \\ {\bf 0} & i\textbf{A}_1' \end{bmatrix}. \nonumber
\end{align*}
Now consider the \mbox{$24 \times 20$} $\textbf{G}$ matrix corresponding to the randomly generated channel realization
\begin{equation*}
\textbf{H} = \begin{bmatrix}[r]
  -0.0583 + 1.2105i  & 0.0708 + 0.6795i\\
  -1.3669 - 0.1373i  &-0.3850 + 0.0877i\\
  -0.3104 - 1.5120i  & 0.2146 + 1.0159i\\
  -1.2690 - 0.5937i  &-0.4245 - 1.3866i\\
   0.5942 + 0.9578i  & 0.3465 - 0.1398i\\
  -0.6279 - 0.7581i  & 0.5228 - 0.8541i     
    \end{bmatrix}.
\end{equation*}
The rank of $\textbf{G}$ is only $18$ and the structure of the $24 \times 20$ upper triangular matrix $\textbf{R}'$ obtained from the \textbf{QR} decomposition is given in~\eqref{eq:R_N6} at the top of the next page. The matrix $\textbf{R}'$ is of the form $\begin{bmatrix}[l] \textbf{R}_1 & {\bf 0}_{9 \times 10} \\ {\bf 0}_{1 \times 10} & {\bf 0}_{1 \times 10} \\ {\bf 0}_{9 \times 10} & \textbf{R}_2 \\ {\bf 0}_{5 \times 10} & {\bf 0}_{5 \times 10} \end{bmatrix}$, where \mbox{$\textbf{R}_1,\textbf{R}_2 \in \mathbb{R}^{9 \times 10}$}. The \mbox{$10 \times 10$} all zero submatrix at the upper right corner of $\textbf{R}'$ is due to the $2$-group ML decodability property of the code. The two sphere decoders corresponding to the two ML decoding groups use the matrices $\textbf{R}_1$ and $\textbf{R}_2$ respectively. Clearly the rank of $\textbf{R}_1$ and $\textbf{R}_2$ is $9$, and hence the sphere decoding complexity is $\mathcal{O}(q)$.
\begin{figure*}
\begin{equation} \label{eq:R_N6}
\textbf{R}' = \begin{bmatrix}
a&a&a&a&a&a&a&a&a&a&0&0&0&0&0&0&0&0&0&0\\
0&a&a&a&a&a&a&a&a&a&0&0&0&0&0&0&0&0&0&0\\
0&0&a&a&a&a&a&a&a&a&0&0&0&0&0&0&0&0&0&0\\
0&0&0&a&a&a&a&a&a&a&0&0&0&0&0&0&0&0&0&0\\
0&0&0&0&a&a&a&a&a&a&0&0&0&0&0&0&0&0&0&0\\
0&0&0&0&0&a&a&a&a&a&0&0&0&0&0&0&0&0&0&0\\
0&0&0&0&0&0&a&a&a&a&0&0&0&0&0&0&0&0&0&0\\
0&0&0&0&0&0&0&a&a&a&0&0&0&0&0&0&0&0&0&0\\
0&0&0&0&0&0&0&0&a&a&0&0&0&0&0&0&0&0&0&0\\
0&0&0&0&0&0&0&0&0&0&0&0&0&0&0&0&0&0&0&0\\
0&0&0&0&0&0&0&0&0&0&a&a&a&a&a&a&a&a&a&a\\
0&0&0&0&0&0&0&0&0&0&0&a&a&a&a&a&a&a&a&a\\
0&0&0&0&0&0&0&0&0&0&0&0&a&a&a&a&a&a&a&a\\
0&0&0&0&0&0&0&0&0&0&0&0&0&a&a&a&a&a&a&a\\
0&0&0&0&0&0&0&0&0&0&0&0&0&0&a&a&a&a&a&a\\
0&0&0&0&0&0&0&0&0&0&0&0&0&0&0&a&a&a&a&a\\
0&0&0&0&0&0&0&0&0&0&0&0&0&0&0&0&a&a&a&a\\
0&0&0&0&0&0&0&0&0&0&0&0&0&0&0&0&0&a&a&a\\
0&0&0&0&0&0&0&0&0&0&0&0&0&0&0&0&0&0&a&a\\
0&0&0&0&0&0&0&0&0&0&0&0&0&0&0&0&0&0&0&0\\
0&0&0&0&0&0&0&0&0&0&0&0&0&0&0&0&0&0&0&0\\
0&0&0&0&0&0&0&0&0&0&0&0&0&0&0&0&0&0&0&0\\
0&0&0&0&0&0&0&0&0&0&0&0&0&0&0&0&0&0&0&0\\
0&0&0&0&0&0&0&0&0&0&0&0&0&0&0&0&0&0&0&0
     \end{bmatrix}
\end{equation}
\hrule
\end{figure*}
\end{example}

\subsubsection*{Non delay-optimal codes}

The non delay-optimal code for \mbox{$N=n2^{\lfloor \frac{g-1}{2}\rfloor}$} has rate \mbox{$R=\frac{N}{2^{g-1}} + \frac{g-1}{2N}$} and number of symbols per group \mbox{$\frac{K}{g}=n^2+g-1$}. The weight matrices of the first group are of the form $\begin{bmatrix} \textbf{D}_1^T, \textbf{D}_2^T, \dots, \textbf{D}_g^T \end{bmatrix}^T$, where each \mbox{$\textbf{D}_j \in \mathbb{C}^{nm \times nm}$}. The first block $\textbf{D}_1$ is one of the $n^2$ matrices of the form $\textbf{V} \otimes \textbf{U}_1$, where \mbox{$\textbf{V} \in \mathbb{C}^{n \times n}$} is Hermitian, and the remaining $g-1$ matrices $\textbf{D}_j$, \mbox{$j=2,\cdots,g$} are of the form \mbox{$\pm \textbf{I}_n \otimes \textbf{U}_j$} for some set of $g$ unitary \mbox{$m \times m$} matrices \mbox{$\textbf{U}_1,\dots,\textbf{U}_g$}.

Using an argument similar to the one used with delay-optimal codes, it can be shown that the sphere decoding complexity of the non delay optimal codes is of the order of $q^{\left( \left( n - 2^{\lfloor \frac{g-1}{2}\rfloor} M \right)^+\right)^2}$ and that the codes are non-singular only for \mbox{$M \geq \frac{N}{2^{g-2}} \approx 2R$}. For $g=2$ and equal values of $N$ and $T$, the non-delay optimal codes of~\cite{NaR1} and the codes of~\cite{RYGGZ} have the same rate and sphere decoding complexity.

\section{Discussion} \label{sec4}

In this paper we have introduced the notion of singularity of STBCs and showed that all known families of high rate multigroup ML decodable codes are singular for certain number of receive antennas. The following facts which were not known before have been shown.
\begin{itemize}
\item Though the $N=4$, $T=4$ code of~\cite{RGYS} and the $N=4$, $T=8$ code of~\cite{RYGGZ} have identical rate of $\frac{17}{8}$ cspcu, the sphere decoding complexity of the code from~\cite{RGYS} is less than that of the code from~\cite{RYGGZ}. 
\item For \mbox{$g=2$}, and equal values of $N$, the delay-optimal codes of~\cite{NaR1} and the codes of~\cite{SrR1} have equal rate and the same order of sphere decoding complexity.
\item For equal values of $N$ and $T$, the codes in~\cite{RYGGZ} and the non-delay optimal codes of~\cite{NaR1} have identical rate and sphere decoding complexities.
\end{itemize}

The results and ideas presented in this paper have brought to light the following important open problems.
\begin{itemize}
\item Is there an algebraic criterion that ensures that a code is non singular? For example, is every code with non vanishing determinant also non-singular for all \mbox{$M \geq R$}?
\item Do there exist high rate multigroup ML decodable codes that are non-singular for arbitrary values of $M$?
\item Do there exist singular high rate multigroup ML decodable STBCs with lower sphere decoding complexity than that of the known codes?
\end{itemize}

\section*{Acknowledgment} 
This work was supported  partly by the DRDO-IISc program on Advanced Research in Mathematical Engineering through a research grant, and partly by the INAE Chair Professorship grant to B.~S.~Rajan 

\section*{Appendix A}

\begin{center}
	\textsc{Proof of Theorem~\ref{th:main_theorem}}
\end{center}

From Proposition~\ref{pr:M_less_than_N} it is clear that the theorem is true for \mbox{$M \geq N$}. Thus, we will only consider the case \mbox{$M<N$}. Before giving the proof of Theorem~\ref{th:main_theorem} we present two results which are used in the proof. Let $\textbf{e}_1,\dots,\textbf{e}_N$ be the $N$ columns of the matrix $\textbf{I}_N$.

\begin{proposition} \label{pr:appendix_first}
For any \mbox{$i=1,\dots,N$}, with probability~$1$ (w.p.1), the vector $\textbf{e}_i$ does not belong to the column space of the channel matrix $\textbf{H}$.
\end{proposition}
\begin{proof}
We first prove the result for \mbox{$i=1$}. Let the channel realization \mbox{$\textbf{H} = \begin{bmatrix} \textbf{H}_1 \\ \textbf{H}_2 \end{bmatrix}$}, where \mbox{$\textbf{H}_1 \in \mathbb{C}^{N-M \times M}$} and \mbox{$\textbf{H}_2 \in \mathbb{C}^{M \times M}$}. Now, consider the matrix \mbox{$\bar{\textbf{H}} = \begin{bmatrix}\textbf{I}_{N-M} & \textbf{H}_1 \\ {\bf 0} & \textbf{H}_2 \end{bmatrix} \in \mathbb{C}^{N \times N}$}. We have, \mbox{$det(\bar{\textbf{H}})=det(\textbf{I}_{N-M})\cdot det(\textbf{H}_2)=det(\textbf{H}_2)$} which is non-zero w.p.1~\cite{TsV}. Thus, w.p.1, the columns of $\bar{\textbf{H}}$ are linearly independent over $\mathbb{C}$, and since $\textbf{e}_1$ is the first column of $\bar{\textbf{H}}$, this means that w.p.1, $\textbf{e}_1$ does not belong to the column space of the matrix $\textbf{H}$.

Now consider any \mbox{$i \in \{1,\dots,N\}$}. Let $\textbf{P}$ be an \mbox{$N \times N$} permutation matrix such that $\textbf{Pe}_i = \textbf{e}_1$. Since $\textbf{P}$ is full-ranked, $\textbf{e}_i$ belongs to the column space of $\textbf{H}$ if and only if $\textbf{e}_1$ belongs to the column space of $\textbf{PH}$. Since $\textbf{P}$ is unitary, the distribution of $\textbf{H}$ and $\textbf{PH}$ are one and the same, and hence the probability that $\textbf{e}_1$ belongs to the column space of $\textbf{PH}$ is $0$. Thus, with probability $1$, $\textbf{e}_i$ does not belong to the column space of $\textbf{H}$.
\end{proof}

For a given channel realization $\textbf{H}$, let \mbox{$\mathcal{S} =\{\textbf{z}  \in \mathbb{C}^N | \textbf{z}^H\textbf{H}=\textbf{0}\}$}. Since $\textbf{H}$ is of rank $M$ w.p.1, the dimension of $\mathcal{S}$ over $\mathbb{R}$ is equal to $2(N-M)$ w.p.1. Let \mbox{$\textbf{z} = [z_1~z_2~\cdots~z_N]^T$} and for \mbox{$i=1,\dots,N$}, let \mbox{$\varphi_i:\mathcal{S} \to \mathbb{R}$} be the vector space homomorphism that sends the vector $\textbf{z}$ to the real number $(z_{i})_Q$. We are interested in the dimension of the subspace of $\mathcal{S}$ which is composed of vectors whose $i^{th}$ component is purely real, i.e., in the dimension of $ker(\varphi_i)$.

\begin{proposition} \label{pr:appendix_second}
For $i=1,\dots,N$, the dimension of image of $\mathcal{S}$ under the map $\varphi_i$, \mbox{$dim\left( ker(\varphi_i) \right) = 2(N-M)-1$} w.p.1.
\end{proposition}
\begin{proof}
For any given $\textbf{H}$, \mbox{$\varphi_i(\mathcal{S}) \subset \mathbb{R}$} and hence $dim\left(\varphi_i(\mathcal{S})\right)$ is either $0$ or $1$. Suppose, $dim\left(\varphi_i(\mathcal{S})\right)=0$, then, there is no vector $\textbf{z}$ in $\mathcal{S}$ such that $z_i$ is non-zero because if such a $\textbf{z}$ exists, the vector $iz_i^* \cdot \textbf{z}$ belongs to $\mathcal{S}$ and the imaginary part of its $i^{th}$ component is $|z_i|^2 \neq 0$, and thus $dim\left(\varphi_i(\mathcal{S})\right)=1$, which is a contradiction. Since, for all the vectors in $\mathcal{S}$, the $i^{th}$ component is $0$, we have 
\begin{equation*}
	\mathcal{S} = \{\textbf{z}|\textbf{z}^H\textbf{H}=\textbf{0}\} = \{\textbf{z}|\textbf{z}^H[\textbf{H}~\textbf{e}_i]=\textbf{0}\}.
\end{equation*}
Thus, the dimension of the column space of $\textbf{H}$ and $[\textbf{H}~\textbf{e}_i]$ are the same. This means that $\textbf{e}_i$ belongs to the column space of $\textbf{H}$. From Proposition~\ref{pr:appendix_first}, $\textbf{e}_i$ belongs to the column space of $\textbf{H}$ w.p.$0$ and hence $dim\left(\varphi_i(\mathcal{S})\right)=0$ w.p.$0$. Thus, $dim\left(\varphi_i(\mathcal{S})\right)=1$  w.p.1. From rank-nullity theorem, \mbox{$dim\left( ker(\varphi_i) \right) = dim(\mathcal{S}) - dim(\varphi_i(\mathcal{S})) = 2(N-M)-1$} w.p.1.
\end{proof}

\vspace{4mm}
\noindent
\emph{Proof of Theorem~\ref{th:main_theorem}:} Let the weight matrices of the STBC $\mathcal{C}_N^{Herm}$ be \mbox{$\textbf{A}_1,\dots,\textbf{A}_{N^2}$} and let the space of \mbox{$N \times N$} Hermitian matrices over $\mathbb{R}$ be given by \mbox{$\mathcal{U} = \langle \textbf{A}_1,\dots,\textbf{A}_{N^2} \rangle$}. For a given channel realization $\textbf{H}$, let \mbox{$\rho: \mathcal{U} \to \mathbb{C}^{N \times M}$} be the $\mathbb{R}$-vector space homomorphism that sends the matrix $\textbf{A}$ to $\textbf{A}\textbf{H}$. Clearly, $rank(\textbf{G})$ is equal to the dimension of the subspace $\rho(\mathcal{U})$ over $\mathbb{R}$. Since $\rho(\mathcal{U})$ is isomorphic to $\mathcal{U}/ker(\rho)$ as vector spaces, we have
\begin{equation*}
rank(\textbf{G}) = dim(\mathcal{U}) - dim(ker(\rho)) = N^2 - dim\left( ker(\rho) \right).
\end{equation*}
Thus, it is enough to show that \mbox{$dim\left(ker(\rho)\right)=(N-M)^2$} w.p.1.
Let \mbox{$\textbf{A} \in ker(\rho)$} and let \mbox{$\textbf{a}_1^H,\dots,\textbf{a}_N^H$} denote the $N$ rows of $\textbf{A}$. Then, $\textbf{a}_1$ satisfies $\textbf{a}_1^H\textbf{H}= \textbf{0}$, and since the $\textbf{A}$ is Hermitian, the first component of $\textbf{a}_1$ is purely real. From Proposition~\ref{pr:appendix_second}, \mbox{$\textbf{a}_1 \in ker(\varphi_1)$} whose dimension is $2(N-M)-1$ w.p.1. Given a choice of $\textbf{a}_1$, since $\textbf{A}$ is Hermitian, the first component of $\textbf{a}_2$ equals the conjugate of the second component of $\textbf{a}_1$, and the second component of $\textbf{a}_2$ is purely real. As a result of these restrictions and since \mbox{$\textbf{a}_2^H\textbf{H}=\textbf{0}$}, i.e., $\textbf{a}_2 \in ker(\varphi_2)$, $\textbf{a}_2$ belongs to a coset of the subspace of $\mathcal{S}$ whose dimension is \mbox{$(2(N-M)-1)-2 = 2(N-M)-3$}. Similarly, given a choice for \mbox{$\textbf{a}_1,\dots,\textbf{a}_{k-1}$}, where \mbox{$k=1,\dots,N$}, the first $k-1$ components of $\textbf{a}_k$ are fixed and the imaginary part of the $k^{th}$ component is zero. Hence, $\textbf{a}_k$ belongs to a coset of a subspace of $\mathcal{S}$ with dimension \mbox{$(2(N-M)-1)-2k$}. Thus, the dimension of $ker\left(\rho\right)$ equals 
\begin{equation*}
2(N-M)-1~+~2(N-M)-3~+~\cdots+~1 = (N-M)^2,
\end{equation*}
with probability $1$. This completes the proof.\\
$~~~~~~~~~~~~~~~~~~~~~~~~~~~~~~~~~~~~~~~~~~~~~~~~~~~~~~~~~~~~~~~~~~~~~~~~~~~~~~~~\blacksquare$

\end{document}